\documentclass[12pt, final]{colt2019}

\usepackage{jmlrutils}
\usepackage{times}

\usepackage{amsmath, amssymb, mathtools, etoolbox}
\usepackage[utf8]{inputenc}
\usepackage[T1]{fontenc}
\usepackage{mathdots}
\usepackage{enumitem}

\newtheorem{algo}[theorem]{Algorithm}

\title{Robustness of spectral methods for community detection}
\coltauthor{%
  \Name{Ludovic Stephan} \Email{ludovic.stephan@inria.fr}\\
  \addr Inria, \'Ecole Normale Supérieure, PSL Research University
  \AND
  \Name{Laurent Massoulié} \Email{laurent.massoulie@inria.fr}\\
  \addr Inria, \'Ecole Normale Supérieure, PSL Research University, Microsoft Research-Inria Joint Centre
}

\providecommand{\given}{}
\DeclarePairedDelimiterXPP{\Pb}[1]{\mathbb{P}}(){}{\renewcommand{\given}{\nonscript\:\delimsize\vert\nonscript\:\mathopen{}} #1}
\DeclarePairedDelimiterXPP{\E}[1]{\mathbb{E}}[]{}{\renewcommand{\given}{\nonscript\:\delimsize\vert\nonscript\:\mathopen{}} #1}
\DeclarePairedDelimiterX{\Set}[1]\lbrace\rbrace{\renewcommand{\given}{\nonscript\:\delimsize\vert\nonscript\:\mathopen{}} #1}
\DeclarePairedDelimiterXPP{\ind}[1]{\mathbf{1}}\lbrace\rbrace{}{#1}
\DeclarePairedDelimiterX{\norm}[1]\lVert\rVert{\ifblank{#1}{\:\cdot\:}{#1}}
\newcommand{\bilingualcommand}[3]{%
	\newcommand{#1}[1][\ ]{%
		##1%
		\iflanguage{english}{\text{#2}}{%
			\iflanguage{french}{\text{#3}}{}%
		}%
		##1%
	}%
}

\bilingualcommand{\where}{where}{où}
\bilingualcommand{\textif}{if}{si}
\bilingualcommand{\textand}{and}{et}
\bilingualcommand{\textiff}{if and only if}{si et seulement si}
\bilingualcommand{\otherwise}{otherwise}{sinon}
\renewcommand{\l}{\ensuremath{^{(\ell)}}}
\newcommand{\ov}{\ensuremath{\mathrm{ov}}}
\renewcommand{\labelenumi}{(\roman{enumi})}

\begin{document}

\maketitle
\begin{abstract}
The present work is concerned with community detection. Specifically, we consider a
random graph drawn according to the stochastic block model: its vertex set is partitioned
into blocks, or communities, and edges are placed randomly and independently of each other with
probability depending only on the communities of their two endpoints. In this context, our aim is
to recover the community labels better than by random guess, based only on the observation of the
graph.

In the sparse case, where edge probabilities are in $O(1/n)$, we introduce a new spectral method
based on the distance matrix $D\l$, where $D\l_{ij} = 1$ iff the graph distance between $i$ and $j$,
noted $d(i, j)$ is equal to $\ell$. We show that when $\ell \sim c\log(n)$ for carefully chosen $c$,
the eigenvectors associated to the largest eigenvalues of $D\l$ provide enough information to
perform non-trivial community recovery with high probability, provided we are above the so-called
Kesten-Stigum threshold. This yields an efficient algorithm for community detection, since
computation of the matrix $D\l$ can be done in $O(n^{1+\kappa})$ operations for a small constant
$\kappa$. 

We then study the sensitivity of the eigendecomposition of $D\l$ when we allow an adversarial
perturbation of the edges of $G$. We show that when the considered perturbation does not affect
more than $O(n^\varepsilon)$ vertices for some small $\varepsilon > 0$, the highest eigenvalues and their
corresponding eigenvectors incur negligible perturbations, which allows us to still perform efficient
recovery.

Our proposed spectral method therefore: i) is robust to larger perturbations than prior spectral methods, while semi-definite programming (or SDP) methods can tolerate yet larger perturbations; ii) achieves non-trivial detection down to the KS threshold, which is conjectured to be optimal and is beyond reach of existing SDP approaches; iii) is faster than SDP approaches.
\end{abstract}

\newpage

\section{Introduction}

\subsection{Background}

Community detection is the task of finding large groups of similar items inside a large relationship
graph, where it is expected that related items are (in the assortative case) more likely to be
linked together.  The Stochastic Block Model (abbreviated in SBM) has been designed by~\cite{Hol83} to analyze the performance of algorithms for this task; it consists in a random
graph $G$ whose edge probabilities depend only on the community membership of their endpoints.
Since then, a large number of articles have been devoted to the study of this model; a survey of
these results can be found in~\cite{Abb17}, or in~\cite{For10} for a more general
view on community detection.

The sparse case, when edge probabilities are in $O(1/n)$, is known to be much harder to study than
denser models; the existence of a positive portion of isolated vertices makes complete
reconstruction impossible, and studies usually focus on partial recovery of the community structure.
Insights on this topic often stem from statistical physics; in the two-community case,~\cite{Dec11} conjectured the existence of a threshold for reconstruction, as well as its exact value; this conjecture was then
proved in~\cite{Mos15} for the first part,~\cite{Mas13} and~\cite{Mos13}
for the converse part. Similarly, in the general case, a method was first presented in~\cite{Krz13} and then proven to work in~\cite{Bor15} -- bar a technical
condition -- and~\cite{Abb16}.

The main issue in the sparse setting is that the usual method relying on the eigenvectors of the adjacency matrix
of $G$ fails due to the lack of separation of the eigenvalues. Consequently, a wide array of
alternative spectral methods have been designed, relying on the spectrum of a matrix associated to
$G$. More precisely, the eigenvectors associated to the highest eigenvalues of those matrices will often carry some
information about the community structure of $G$, enough for partial reconstruction. Examples
include the path expansion matrix used in~\cite{Mas13}, or the non-backtracking matrix
in~\cite{Krz13}.

Additionally, other types of methods can be used in this setting: for example, the semi-definite
programming (or SDP) algorithm relaxes the problem into a convex optimization one, which can then be
approximately solved (see for example~\cite{Mon16}).

An important feature of real-life networks that is missing from the SBM is the existence of
small-scale regions of higher density, that arise from phenomena unrelated to the community
structure. For this reason, a common variant of the SBM is the addition of small cliques to the
generated random graph. Commonly-used spectral methods, for example those relying on the
non-backtracking matrix in~\cite{Bor15}, are known to fail in this setting, due to the apparition of
localized eigenvectors, with no ties to the community structure, and corresponding to large
eigenvalues -- see~\cite{Zha16} for a comparison of those methods, as well as a proposed heuristic to
deal with those localized vectors by lowering their associated eigenvalues. SDP methods are the most studied for this problem, due to their natural stability; in particular, ~\cite{Mak16} show a reconstruction algorithm that is robust to the adversarial addition of $o(n)$ edges, in the case of an arbitrary number of communities; this was also shown independently by~\cite{Moi16}. However, all the SDP methods mentioned here fail to reach the KS threshold by at least a large constant, with only~\cite{Mon16} approaching it as the average degree increases.

After completing this work we became aware of the article of \cite{Abb18}. It establishes results akin to ours on robustness (although with a different definition thereof) of spectral methods for detection in SBM. We use however a slightly different matrix, and our results apply to an arbitrary number of blocks, whereas they only consider SBM with two blocks.

\subsection{Summary of main results}
This article focuses on the Stochastic Block Model, as defined in \cite{Hol83}; we recall here a succinct definition:

\begin{definition}
  Let $r \in \mathbb N$ be fixed, $W$ be a $r\times r$ symmetric matrix with nonnegative entries, and $\pi$ a probability vector on $[r]$. A random graph $G = (V, E)$ with $|V| = n$ is said to be distributed according to the Stochastic Block Model (or SBM) with $r$ blocks and parameters $(W, \pi)$ if:
  \begin{enumerate}
      \item each vertex $v \in V$ is assigned a \emph{type} $\sigma(v)$ sampled independently from $\pi$,
      \item any two vertices $u, v$ in $V$ are joined with an edge randomly and independently from every other edge, with probability
      \[ \min(\frac{W_{\sigma(u), \sigma(v)}}n, 1). \]
  \end{enumerate}
\end{definition}

Given a random graph $G$ sampled according to the above model (with the types of each vertex hidden), the aim is to estimate the type assignment $\sigma$ from the observation of $G$ only. However, since there is a positive proportion of isolated vertices, perfect reconstruction is theoretically impossible; we will thus only focus on retrieving only a positive proportion of the types. Our metric for reconstruction will be the following:
\begin{definition}
  Let $\sigma$ be the true type estimation, and $\hat\sigma$ a type estimate of $\sigma$; the \emph{empirical overlap} between $\sigma$ and $\hat\sigma$ is defined as:
  \begin{equation}
  \ov(\sigma, \hat\sigma) = \max_{\tau\in \mathfrak S_r}\ \frac 1 n \sum_{v=1}^n{\mathbf{1}_{\hat\sigma(v) = \tau(\sigma(v))}}
  - \max_{k\in [r]} \pi_k,
\end{equation}
where $\mathfrak S_r$ is the set of permutations of $[r]$.
For a given algorithm leading to estimates $\hat\sigma$ for all $n$, we say that this algorithm achieves partial reconstruction if \begin{equation}
  \liminf_{n\to\infty}\ \ov(\sigma, \hat\sigma) > 0 \quad\text{w.h.p}.
\end{equation}
\end{definition}

Spectral methods in denser settings (with average degrees of about $\log(n)$) usually consist in clustering the eigenvectors of the adjacency matrix of $G$; however those methods are known to fail in sparser graphs (see \cite{Abb17}. As a result, different (and more complex) matrices are needed:
\begin{definition}
  Let $G$ be any graph, and $\ell$ be a positive integer. We define two matrices associated with $G$:
  \begin{enumerate}
      \item the \emph{path expansion matrix} $B\l$ (studied in \cite{Mas13}), whose $(i, j)$ coefficient counts the number of self-avoiding paths (that is, paths that do not go through the same vertex twice) of length $\ell$ between $i$ and $j$,
      \item the \emph{distance matrix} $D\l$, defined by $D\l_{ij} = 1$ if $d(i, j) = \ell$ and 0 otherwise, where $d$ denotes the usual graph distance.
  \end{enumerate}
\end{definition}

We are now ready to state our first result:
\begin{theorem}\label{thm:summary-algo}
  Assume that $\pi \equiv 1/r$, $W$ is a stochastic positive regular matrix, and that the two highest eigenvalues $\mu_1, \mu_2$ of $W$ satisfy the condition:
  \[ \mu_2^2 > \mu_1. \]
  Then there exists an algorithm, based only on an eigenvector of $B\l$ associated with its second highest eigenvalue, that achieves partial reconstruction whenever $\ell \sim \delta \log(n)$ for small enough $\delta$.
  
  The same algorithm also achieves partial reconstruction when applied to $D\l$ instead of $B\l$, with the same conditions on $\ell$.
\end{theorem}

Regardless of the change of matrices, this is already an improvement on \cite{Bor15}; indeed, we managed to remove a technical asymmetry condition on $W$ (namely, the existence of a simple eigenvector associated with a high eigenvalue).

\vspace{1em}

We now move on to study the stability of our algorithm; as opposed to most papers that classify the difficulty of an adversary according to the number of altered edges, ours considers the number of affected vertices.
\begin{definition}
  Let $\gamma := \gamma(n)$ be a positive integer, and $G$ any graph on $n$ vertices. The adversary of \emph{strength} $\gamma$ is allowed to arbitrarily add and remove edges at will, as long as the number of vertices affected (i.e. vertices that are endpoints of altered edges) is at most $\gamma$. 
\end{definition}

Our main result on stability is then the following:
\begin{theorem}
    Under the same assumptions as Theorem~\ref{thm:summary-algo}, let $G$ a graph generated via SBM, and $\tilde G_\gamma$ a graph obtained when perturbed by an adversary of strength $\gamma$.\\
    Then, assuming
    \[ \gamma = o\left(\frac{(\mu_2^2/\mu_1)^{\ell/2}}{\log(n)}\right), \]
    the algorithm of Theorem~\ref{thm:summary-algo} still achieves partial reconstruction on $\tilde G_\gamma$.
    
    The above result on $\gamma$ is the best possible, up to a factor of $\log(n)$.
\end{theorem}

Compared to the spectral algorithm in \cite{Bor15}, this is a substantial improvement: their algorithm is known (see e.g. \cite{Zha16}) to be highly unstable w.r.t edge addition. In contrast, the above result reaches a perturbation of size a small power of $n$ (since $\ell$ is of order $\log(n)$). This is sharp, and thus still far from the $o(n)$ bound achieved by various SDP methods (notably \cite{Mak16,Mon16}); this discrepancy is likely a result of delicate graph properties involved in spectral algorithms that make them more sensitive to perturbations.

However, our result still has several advantages compared to the other cited methods, namely:
\begin{enumerate}[topsep=2pt, partopsep=0pt, itemsep=0pt]
  \item the threshold for partial reconstruction in our method is exactly the KS threshold,
    whereas SDP-based methods require a slightly stronger condition, especially when the mean
    degree of $G$ is low.
  \item as will be proved later, the running time of our algorithm is at most $O(n^{13/12})$, which is much faster that the usual methods for SDP algorithms.
  \item Finally, all the SDP methods mentioned throughout this paper only consider the symmetric case even in the case of multiple communities.
\end{enumerate}

\subsection{Detailed setting and results}
\subsubsection*{Stochastic block model}
Consider the SBM as defined above; following~\cite{Bor15}, we introduce $\Pi = \mathrm{diag}(\pi_1, \ldots, \pi_r)$ and define the
\emph{mean progeny matrix} $M = \Pi W$; the eigenvalues of $M$ are the same as those of the symmetric matrix
$S = \Pi^{1/2}W\Pi^{1/2}$ and in particular are real. We denote them by
\[ \mu_1 \geq |\mu_2| \geq \ldots \geq |\mu_r|. \]

We shall make the following regularity assumptions: first, 
\begin{equation}
  \mu_1 > 1 \quad\textand M \text{ is positive regular},
\end{equation}
i.e.\ the coefficients of $M^t$ are all positive for some $t$. Secondly, each type of vertex
has the same asymptotic average degree, that is
\begin{equation}\label{degree}
  \sum_{i=1}^r{M_{ij}} = \sum_{i=1}^r{\pi_i W_{ij}} = \alpha\quad \text{for all } j\in[r].
\end{equation}

In this case, the matrix $M^* = M/\alpha$ is a stochastic matrix and therefore 
\begin{equation}
  \mu_1 = \alpha > 1.
\end{equation}

Since $M = \Pi^{-1/2}S\Pi^{1/2}$, M is diagonalizable; let $(\phi_1, \ldots, \phi_r)$ be a basis of
normed left eigenvectors for $M$, that is
\begin{equation}
   \phi_i^\top M = \mu_i \phi_i^\top \quad \text{for all } i\in[r].
\end{equation}

Condition (\ref{degree}) implies that $\phi_1 = \mathbf{1}/\sqrt{r}$, where $\mathbf 1$ is the all-ones
vector.

It has been shown in~\cite{Bor15} and~\cite{Abb16} that polynomial-time algorithms achieve partial
reconstruction when the following condition, called the Kesten-Stigum threshold, is verified:
\begin{equation}\label{ks}
  \tau := \mu_2^2 / \mu_1 > 1.
\end{equation}
The quantity $\tau$ is commonly referred to as the \emph{signal-to-noise ratio}.

Alternatively, we define $r_0$ such that
\begin{equation}
\mu_{r_0 + 1}^2 \leq \mu_1 < \mu_{r_0}^2. 
\end{equation}
Therefore, the condition~(\ref{ks}) is equivalent to $r_0 > 1$.

In the two-community case, the above condition is equivalent to the possibility of reconstruction
(see~\cite{Mas13},~\cite{Mos15}). However, in the general setting ($r > 4$), non-polynomial algorithms can achieve 
partial reconstruction even below this threshold. This was originally conjectured in~\cite{Dec11},
and more recently proven in~\cite{Abb16}.
\subsubsection*{Path expansion matrix}

In~\cite{Mas13}, an algorithm for partial reconstruction in the two-community case makes use of the
path expansion matrix $B\l$. Our first aim is to extend the result from this paper to the general case; we first define for all
$k\in[r]$ the vectors $\chi_k$ and $\varphi_k$ by
\begin{equation}\label{eigen-definition}
\chi_k(v) = \phi_k(\sigma(v))\quad\textand\quad\varphi_k = \frac{B\l\chi_k}{\norm{B\l\chi_k}}.
\end{equation}

Let $\lambda_1(B\l) \geq |\lambda_2(B\l)| \geq |\lambda_n(B\l)|$ be the eigenvalues of $B\l$ ordered
by absolute value; our first theorem is an extension of Theorem 2.1 in~\cite{Mas13}:

\begin{theorem}\label{path-expansion}
  Consider a graph $G$ generated as above, and let $\ell \sim \kappa\log_\alpha(n)$, with $\kappa < 1/12$.
  Then, with probability going to $1$ as $n$ goes to $+\infty$:
  \begin{enumerate}
    \item $\lambda_k(B\l) = \Theta(\mu_k^\ell)$ for $k\in[r_0],$ 
    \item For $k > r_0$, $\lambda_k(B\l) = O(\log(n)^c\,\alpha^{\ell/2})$ for some constant $c > 0$.
  \end{enumerate}
  Furthermore, consider $\mu$ such that $\mu^2 > \alpha$ and $\mu$ is an eigenvalue of multiplicity $d$ of M. Let $\phi^{(1)}, \dots, \phi^{(d)}$ be an orthonormal basis of eigenvectors of $M$ associated to $\mu$, and $\varphi^{(1)}, \dots, \varphi^{(d)}$ the vectors defined as in (\ref{eigen-definition}).\\
  There exist orthogonal vectors $\xi^{(1)}, \dots, \xi^{(d)}$ in $\mathbb R^n$ such that the following holds:
  \begin{enumerate}
    \item for all $i$, $\xi^{(i)}$ is an eigenvector of $B\l$, with associated eigenvalue $\Theta(\mu^\ell)$
    \item there exists an orthogonal matrix $Q \in O(d)$ such that 
    $$\norm{\boldsymbol\varphi Q - \boldsymbol\xi}_2 = O\left(\alpha^{\ell/2}\mu^{-\ell}\right),$$ 
    where $\boldsymbol\varphi$ (resp. $\boldsymbol\xi$) is the $n\times d$ matrix whose columns are the $\varphi^{(i)}$ (resp. the $\xi^{(i)}$).
  \end{enumerate}
\end{theorem}

The above theorem does not yield immediately an algorithm for community reconstruction; however, adapting the one found in~\cite{Bor15}, we designed the following:

\begin{algo}\label{algo-reconstruction}
  Let $\xi$ be an eigenvector of $B\l$ associated to the eigenvalue $\lambda_2(B\l)$, normalized such that $\norm{\xi}^2 = n$, and $K$ an arbitrary large constant. 
  First, partition $V$ in two sets $(I^+, I^-)$ via the following procedure: put $v$ in $I^+$ with probability
    $$ \Pb{v \in I^+ \given \xi} = \frac12 + \frac1{2K}\,\xi(v)\,\mathbf 1_{|\xi(v)| \leq K} $$

  Then, assign the label 1 to every vertex in $I^+$ and label 2 to every vertex in $I^-$.
\end{algo}

We then have the following theorem:

\begin{theorem}\label{algorithm}
  Assume that $\pi \equiv 1/r$, and that $r_0 > 1$, i.e. that we are above the Kesten-Stigum threshold. Then Algorithm~\ref{algo-reconstruction} yields an asymptotically positive overlap when $n\to\infty$ for some choice of $K$.
\end{theorem}

Note that we don't need the asymmetry condition from \cite{Bor15} anymore; our algorithm can deal with multiple eigenvalues as well. Additionally, an explicit value for $K$ is derived in the appendix, which makes our algorithm easy to implement and eliminates the need for ``magic'' constants, such as the ones in \cite{Zha16} or \cite{Abb16}.

A crucial feature of this algorithm is that it depends only on the second eigenvalue of $B\l$; for any perturbation that leaves the $r_0$ highest eigenvalues -- or even the second highest -- unchanged, the result from Theorem~\ref{algorithm} will hold.

\subsubsection*{The distance matrix}

We introduce now the \emph{distance matrix} $D\l$, defined by $D\l_{ij} = 1$ if and only if
$d(i,j) = \ell$, where $d$ is the distance in $G$. This matrix, while sparser than $B\l$, retains
much of the desired spectral properties. In particular, we have the following theorem:

\begin{theorem}\label{distance}
  Assume that condition~(\ref{ks}) holds, and set $\ell$ such that $\ell\sim \kappa\log_\alpha(n)$, where
  $\kappa$ is a constant such that $\kappa < 1/12$. Then the results of Theorem~\ref{path-expansion} hold with the matrix $B\l$ replaced by $D\l$.
\end{theorem}

As a result, Algorithm~\ref{algo-reconstruction} will still succeed when applied to the
matrix $D\l$.

\subsubsection*{Graph perturbation}

As mentioned in the introduction, community detection algorithms have to be resilient to the
presence of small cliques (or denser subgraphs) to be useful in practice, since this kind of pattern
is often present in real-life networks. We chose to focus here on adversarial perturbations, as defined in the summary, whereas other papers (mainly \cite{Abb18}) focus instead on other random graph models, more prone to small loops and cliques.

As shown in~\cite{Zha16}, the usual spectral methods do not fare well against adversarial (or even random)
perturbation, especially when the added subgraph contains several cliques. This is especially the
case for the non-backtracking matrix in~\cite{Bor15}, but also the path expansion matrix
in~\cite{Mas13}.

However, the distance matrix is more stable to clique addition, since it does not count the number
of paths between two vertices -- which is affected significantly by small perturbations. We can therefore
allow a perturbation of size up to a small power of $n$, as stated in the following theorem:

\begin{theorem}\label{generalpert}
  Let $G$ be an SBM as above, with $\pi_i \equiv 1/r$. Assume that $r_0 > 1$, and recall that $\tau = \mu_2^2 / \mu_1 >
  1$ is the signal-to-noise ratio.

  Then, if $\gamma = o(\tau^\ell / \log(n))$,
  then Algorithm \ref{algo-reconstruction} based on the distance matrix recovers the original communities with
  asymptotically positive overlap, even after a perturbation affecting at most $\gamma$ vertices.
\end{theorem}

The controls in the above theorem can be shown to be sharp, up to a factor of $\log(n)$:

\begin{theorem}\label{optimality}
  With the same assumptions as above, let $D\l$ be the distance matrix of $G$ and $\tilde D\l$ the
  one of the graph after the adversarial perturbation.

  If $\gamma = \Omega(\tau^\ell)$, then there exists a perturbation of size at most $\gamma$
  such that $\tilde D\l$ has an eigenvalue of size $\Omega(\mu_2^\ell)$, with associated eigenvector
  asymptotically perpendicular to the first $r_0$ ones of $D\l$.
\end{theorem}

Therefore, we cannot guarantee the stability of the eigenvectors of $D\l$ when the perturbation
affects too many vertices. This means that the best bound we can get on the size of allowed perturbations of the matrix $D\l$ is $\tau^\ell$, which we can rewrite as
$$\tau^\ell = n^{\kappa\log_\alpha(\tau)}.$$

The spectral method on the distance matrix is thus robust to perturbations of size at most $n^\varepsilon$, with $\varepsilon = \kappa \log_\alpha(\tau)$ going to zero as we approach the KS threshold.

\subsection{Notations and outline of the paper}

Throughout this paper, we will make use of the following notation: for two functions $f, g$, we say
that $f = \widetilde{O}(g)$ if there exists a constant $c$ such that $f = O\left(\log(n)^c \cdot
g\right)$. We similarly define the notations $\widetilde\Theta$ and $\widetilde\Omega$.

The next Section is devoted to the study of the spectral structure of $B\l$; we also state there an
important theorem on spectral perturbation that will be useful for the study of matrix $D\l$ as
well. In Section 3, we study the distance matrix $D\l$ and introduce a method to deal with
perturbations of this matrix. We then leverage this method to obtain bounds on the size of allowed
perturbations.

\section{Spectral structure of $B\l$}

\subsection{A theorem on eigenspace perturbation}

In the following, we'll need a way to link the operator norm of a matrix perturbation to the
consequent perturbation of its eigenvectors. This is provided by the following variant of the
Davis-Kahan $\sin\theta$ theorem (\cite{Yu15}, Theorem 2):

\begin{theorem}\label{perturbation}
  Let $\Sigma, \hat\Sigma$ be symmetric $n\times n$ matrices, with eigenvalues $\lambda_1 \geq
  \ldots \geq \lambda_n$ and $\hat\lambda_1 \geq \ldots \geq \hat\lambda_n$ respectively.
  Fix $1\leq r \leq s \leq n$ and assume that $\min(\lambda_{r-1} - \lambda_r, \lambda_s -
  \lambda_{s+1}) > 0$, where we define $\lambda_0 = +\infty$ and $\lambda_{n+1} = -\infty$.

  Let $d = s - r + 1$, and let $V = (v_r, \ldots, v_s)$ and $\hat V = (\hat v_r, \ldots, \hat v_s)$
  have orthonormal columns satisfying $\Sigma v_j = \lambda_j v_j$ and $\hat\Sigma \hat v_j =
  \hat\lambda_j \hat v_j$ for $j\in\{r, \ldots, s\}$.
  
  Then there exists an orthogonal matrix $Q \in O(d)$ such that
  \begin{equation}\label{davis-kahan}
    \norm{VQ - \hat V}_F \leq \frac{2\sqrt{2d}\norm{\hat\Sigma - \Sigma}_{\mathrm{op}}}{\min(\lambda_{r-1} - \lambda_r, \lambda_s -
    \lambda_{s+1})}\cdot
  \end{equation}
\end{theorem}

\subsection{Strategy of proof}

We present here the main ideas of the proof, and defer its full version to the appendix. The first
step is an adaptation of Proposition 19 from~\cite{Bor15}:

\begin{proposition}\label{scalar}
  Let $\ell \sim \kappa\log_\alpha(n)$ with $\kappa < 1/12$. Define, for $k \in [r]$, 
  \begin{equation}
    \theta_k = \norm{B\l \varphi_k} \textand \zeta_k = \frac{B\l\varphi_k}{\theta_k},
  \end{equation}
  with $\varphi_k$ as in~(\ref{eigen-definition}).

  Then, with high probability, we have the following estimations for every $\gamma < 1/2$:
  \begin{enumerate}
    \item $\theta_k = \Theta(\mu_k^\ell)$ for $k\in[r_0]$,
    \item $|\langle \varphi_j, \varphi_k\rangle | = \widetilde O(\alpha^{3\ell/2}n^{-\gamma/2})$ for
      $j\neq k \in[r_0]$,
    \item $|\langle \zeta_j, \varphi_k\rangle | = \widetilde O(\alpha^{2\ell}n^{-\gamma/2})$ for
      $j\neq k \in[r_0]$.
  \end{enumerate}
\end{proposition}

Now, let $(z_1, \ldots, z_{r_0})$ be the Gram-Schmidt orthonormalization of $(\varphi_1, \ldots,
\varphi_{r_0})$, and define
\[ D = \sum_{k=1}^{r_0}{\theta_k z_k z_k^\top}. \]

The non-zero eigenvalues of $D$ are thus the $\theta_k$, with corresponding eigenvectors $z_k$.
Then, using the asymptotic orthogonality properties of Proposition~\ref{scalar}, we prove the
following:
\begin{proposition}\label{norm}
  For all $k\in[r_0]$, $z_k$ is asymptotically parallel to $\varphi_k$.

  Furthermore, 
  \begin{equation}
    \norm{B\l - D}_{\mathrm{op}} = \widetilde O(\alpha^{\ell/2}).
  \end{equation}
\end{proposition}

\noindent Theorem~\ref{path-expansion} then results from a simple application of the Weyl inequality
(\cite{Weyl}) and Theorem~\ref{perturbation}:

\begin{proof}(of Theorem~\ref{path-expansion}): let $\theta_k = 0$ for $k > r_0$; the eigenvalues
  of $D$ are then exactly the $\theta_i$ for $i \leq n$.

  By Weyl's inequality, we have for all $i \in[n]$
  \[ |\lambda_i(B\l) - \theta_i| = \widetilde O(\alpha^{\ell/2}). \]
  Since $\theta_k = \Theta(\mu_k^\ell)$ for $k\in[r_0]$, this implies the statements (i) and (ii) of
  the Theorem.

  We now define $z^{(1)}, \dots, z^{(d)}$ as the $z_i$ associated to $\varphi^{(1)}, \dots, \varphi^{(d)}$, and $\mathbf z$ as in Theorem~\ref{path-expansion}. Applying inequality (\ref{davis-kahan}) to $B\l$ and $D$ yields the existence of an orthogonal matrix $Q \in O(d)$ such that
  $$\norm{\mathbf z Q - \boldsymbol \xi} = O\left(\alpha^{\ell/2}\mu^{-\ell}\right),$$

  and the proof of Proposition~\ref{norm} shows that $\norm{z^{(i)} - \varphi^{(i)}} =  O\left(\alpha^{\ell/2}\mu^{-\ell}\right)$ for all $i$. Using the triangular inequality (and the fact that $Q$ preserves the norm) completes the proof of Theorem~\ref{path-expansion}.
\end{proof}

\subsection{A new reconstruction algorithm}

We now sketch the proof for Theorem~\ref{algorithm}; it hinges on one key lemma, whose proof (adapted from~\cite{Bor15}) is in the appendix:

\begin{lemma}\label{lem-conv}
  Let $\xi$ be as in Algorithm~\ref{algo-reconstruction}. For all $i \in [r]$, there exists a random variable $X_i$ such that for every $K > 0$ that is a continuity point of $X_i$, in probability,
  \[ \frac 1 n \sum_{v\in V}{\mathbf 1_{\sigma(v) = i}\,\xi(v)\mathbf 1_{|\xi(v)| \leq K}} \to  \pi(i)\, \E{X_i \mathbf 1_{|X_i| \leq K}}, \]
  where the convergence is independent from $\xi$.

  Furthermore, we have
  \begin{equation}\label{non-constant}
  \sum_{i\in [r]}{\E{X_i}} = 0 \quad \textand \quad \sum_{i\in [r]}{\E{X_i}^2} > c
  \end{equation}
  for some absolute constant $c > 0$, and for all $\varepsilon > 0$ there exists a choice of $M$ (independent from the chosen eigenvector $\xi$) such that
  \begin{equation}\label{bounded-exp}
    \left|\E{X_i \mathbf{1}_{|X_i| \leq K}} - \E{X_i}\right| < \varepsilon 
  \end{equation}
\end{lemma}
In particular, Lemma~\ref{lem-conv} implies that there exist $i$ and $j$ such that $|\E{X_i} - \E{X_j}| > \sqrt{c}$.

We then use a concentration bound to show that for all $i$, in probability,
\begin{equation}\label{eq-proportion}
  \frac1n\sum_{v \in V} \mathbf{1}_{\sigma(v) = i}\,\mathbf{1}_{v \in I^+} \to \pi(i)\left(\frac{\E{X_i \mathbf 1_{|X_i| \leq K}}}{2K} + 1/2\right) := \pi(i) \tilde p_i
\end{equation}
where the convergence is independent from $\xi$.

Assume now that $\pi \equiv 1/r$; from Lemma~\ref{lem-conv}, for a large enough $M$ there exists a $\delta > 0$ such that $\tilde p_i > \tilde p_j + \delta$. Assign label 1 to $I^+$ and 2 to $I^-$, and let $\tau$ be a permutation such that $\tau(i) = 1$ and $\tau(j) = 2$. The overlap achieved by $\tau$ is thus
\begin{equation} 
  \frac1n\sum_{v \in V} \mathbf{1}_{\sigma(v) = i}\,\mathbf{1}_{v \in I^+} + \frac1n\sum_{v \in V} \mathbf{1}_{\sigma(v) = j}\,\mathbf{1}_{v \in I^-} - \frac1r = \frac1r\left(\tilde p_i + 1 - \tilde p_j\right) - \frac1r > \frac{\delta}r,
\end{equation}
which completes the proof of Theorem~\ref{algorithm}.

\section{Study of the matrix $D\l$}

\subsection{From $B\l$ to $D\l$}

The first aim of this section is to prove Theorem~\ref{distance}, i.e. that we can replace
matrix $B\l$ by $D\l$ in the algorithm from Theorem~\ref{algorithm}. Directly proving this theorem is hard, because of the lack of a decomposition such as the one in Lemma~\ref{decomposition} for $D\l$. However, in view of the
proof of Theorem~\ref{path-expansion} above, it is sufficient to prove the following proposition:

\begin{proposition}\label{spectral-radius}
Let $G$ be a SBM as above, and $\ell\sim\kappa\log_\alpha(n)$ with $\kappa < 1/12$. Let $B\l$ be the
path expansion matrix of $G$, and $D\l$ its distance matrix. Then, with high probability:
\begin{equation}
  \rho(B\l - D\l) = \widetilde{O}(\alpha^{\ell/2}),
\end{equation}
where $\rho$ is the spectral radius of a matrix.
\end{proposition}

For ease of notation, let $\Delta\l = B\l - D\l$; we first notice that $\Delta\l$ is a $0-1$ matrix:

\begin{lemma}\label{zero-one}
  Let $\ell\sim\kappa\log_\alpha(n)$ with $\kappa < 1/12$. For all vertices $i,j \in \Set{1,\ldots,n}$,
  \begin{equation}
    0\leq \Delta\l_{ij} \leq 1.
  \end{equation}

  Furthermore, if $\Delta\l_{ij} = 1$, then there exists a cycle $\mathcal{C}$ such that:
  \begin{equation}
    d(i, \mathcal{C}) + d(j, \mathcal{C}) \leq \ell.
  \end{equation}
\end{lemma}

Define now a matrix $P\l$ by $P\l_{ij} = 1$ if there is a cycle $\mathcal{C}$ such that
$d(i, \mathcal{C}) + d(j, \mathcal{C}) \leq \ell$. By the previous lemma, we have $\Delta\l_{ij} \leq P\l_{ij}$ for
all $(i,j)$, and the Perron-Frobenius theorem implies: 
  \begin{equation}\label{spectral-bound}
    \rho(\Delta\l) \leq \rho(P\l).
  \end{equation}
  It remains then to bound the spectral radius of $P\l$; the key lemma is the
following:
\begin{lemma}\label{cycle-decomposition}
  For a given cycle $\mathcal{C}$, let $P\l_\mathcal{C}$ be the matrix defined by $P\l_{\mathcal{C}, ij} = 1$ if
  $d(i, \mathcal{C}) + d(j, \mathcal{C}) \leq \ell$, and $V_\mathcal{C}$ the set of vertices such that
  $d(i, \mathcal{C}) \leq \ell$. Then:
  \begin{enumerate}
    \item $P\l_\mathcal{C}$ is zero outside of $V_\mathcal{C} \times V_\mathcal{C}$,
    \item $\rho(P\l) = \max_\mathcal{C} \rho(P\l_\mathcal{C})$.
  \end{enumerate}
\end{lemma}

By part (ii) of the above lemma, it is sufficient to bound $\rho(P\l_\mathcal{C})$ for a given cycle $\mathcal{C}$
in $\mathcal{G}$; using part (i) of the above lemma, we can restrict our study to the subspace spanned by
the vertices in $V_\mathcal{C}$.

Let $v$ be a normed vector of size $|V_\mathcal{C}|$ corresponding to the highest eigenvalue of $P\l_\mathcal{C}$;
as the coefficient $(i, j)$ of $P\l_\mathcal{C}$ only depends on the distance of $i$ and $j$ to $\mathcal{C}$, we
likewise group the coefficients of $v$ by their distance $t$ to $\mathcal{C}$, and write 
\[ v = {(v_{tj})}_{\substack{0\le t\le \ell\\1\le j\le
S_t(\mathcal{C})}}. \]

We then have: 
\[v^\top P_\mathcal{C} v = \sum_{t + u \leq \ell}{\sum_{i,j}{v_{ti}\,v_{uj}}} = 
\sum_{t+u\leq \ell}{\left(\sum_i{v_{ti}}\right)\left(\sum_j{v_{uj}}\right)}. \]

By the Perron-Frobenius theorem, the coefficients of $v$ are non-negative. For a given $t$, the
coefficients $v_{ti}$ are necessarily equal; otherwise, we could increase $\sum{v_{ti}}$  while
leaving $\sum{v_{ti}^2}$ fixed, which leads to increasing $ v^\top P\l_\mathcal{C} v$ while keeping $\norm{v}^2$
  constant: this contradicts the definition of $v$.

Writing $v_{ti} = v_t$ for all $1 \leq i \leq S_t(\mathcal{C})$; we get:

\begin{equation}
   v^\top P_\mathcal{C} v = \sum_{t+u \leq \ell}{S_t(\mathcal{C})S_u(\mathcal{C})v_t v_u}\quad\textand\quad\norm{v}^2 =
  \sum_t{S_t(\mathcal{C})v_t^2}.
\end{equation}

Let $w$ be the size $\ell$ vector defined by $w_t = \sqrt{S_t(\mathcal{C})}v_t$. Rewriting the above
expression in terms of $w$ yields
\begin{equation}
v^\top P_\mathcal{C} v = \sum_{t+u \leq
\ell}{\sqrt{S_tS_u}w_tw_u}\quad\textand\quad\norm{v}^2 = \norm{w}^2,
\end{equation}
where we omit the dependency of $S_t$ in $\mathcal{C}$.

As a result, the spectral
radius of $P_\mathcal{C}$ is equal to that of the $\ell \times \ell$ matrix $Q_\mathcal{C}$ defined by:
\[ Q_\mathcal{C} = \begin{pmatrix}
   S_0 & \sqrt{S_0S_1} & \cdots & \sqrt{S_0S_{\ell-1}} & \sqrt{S_0S_\ell} \\
   \sqrt{S_0S_1} & S_1 & \cdots & \sqrt{S_1S_{\ell-1}} & 0 \\
   \vdots & \vdots & \iddots & \vdots & \vdots \\
   \sqrt{S_0S_{\ell-1}} & \sqrt{S_1S_{\ell-1}} & \cdots & 0 & 0 \\
   \sqrt{S_0S_{\ell}} & 0 & \cdots & 0 & 0 \\
\end{pmatrix}.\]

We now finally use the row sum bound to get:

\begin{align}
  \rho(P\l_\mathcal{C}) = \rho(Q_\mathcal{C}) & \leq \max_{t}\sum_{u\leq \ell-t}{\sqrt{S_tS_u}} \\
             & \leq \max_{t}\sum_{u\leq
\ell-t}{\log(n)\alpha^{\frac{t+u}2}}\quad\text{via lemma \ref{vertex-size}} \\
             & = O(\log(n)\alpha^{\ell/2}).
\end{align}

Combining the above inequality with Lemma~\ref{cycle-decomposition} and
inequality~(\ref{spectral-bound}) eventually leads to 
\begin{equation}
  \rho(\Delta\l) = \widetilde O(\alpha^{\ell/2}),
\end{equation}

which completes the proof of Proposition~\ref{spectral-radius}.

\subsection{Stability to graph perturbation}\label{section-graphpert}

In this subsection, we sketch the proofs for Theorems~\ref{generalpert} and~\ref{optimality}.

\paragraph{A note about computational complexity}

In the original algorithm, the computation of $B\l$ in polynomial time relies on the almost
tree-like, tangle-free structure of the random graph $\mathcal{G}$; this structure may be lost when we add
cliques, and increase the algorithm complexity. As we want to devise polynomial algorithms in every
case, this may be a hindrance.

Conversely, the computation of the distance matrix $D\l$ can be done in polynomial time (for example
breadth-first search of the $\ell$-neighbourhood of each vertex in $G$ yields an algorithm in
$O(n^{1+\kappa}) = O(n^{13/12})$ in the case of SBM, $O(n^2)$ in general) for any graph, which makes
it all the more adapted to the problem at hand.

\vspace{1em}

In order to prove Theorem~\ref{generalpert}, we need a less restrictive version of
Proposition~\ref{spectral-radius}; indeed, bounding the spectral radius of the perturbation by
$O(\alpha^{\ell/2})$ not only preserves the highest eigenvalues, but also bounds the remaining
eigenvalues of $D\l$ by $\sqrt{\lambda_1(D\l)}$. This bound is commonly referred to as a Ramanujan-like
property of $G$.

This property, although interesting on its own, is not specifically needed for
the reconstruction algorithm to work; rather, we only need one eigenvector associated to the second highest eigenvalue $\mu_2$ to remain unchanged.

We'll therefore only need the following proposition:

\begin{proposition}\label{general-spectral}
  We consider the same setting as Theorem~\ref{generalpert}. Let $D\l$ be the distance matrix of
  $G$, and $\tilde G$ and $\widetilde D\l$ be the perturbed versions (after adding adversarial noise) of $G$ and $D\l$, respectively.

  Then
  \begin{equation}
    \rho(\widetilde D\l - D\l) = o(\mu_2^\ell).
  \end{equation}

\end{proposition}

The proof relies on a bound similar to the one in Theorem~\ref{distance}, replacing matrices $P_\mathcal{C}$
and $Q_\mathcal{C}$ by matrices $P_\mathcal{K}$ and $Q_\mathcal{K}$ also depending only on the distance to the perturbed
vertex set $\mathcal{K}$. The details can be found in the appendix.

\newpage

\bibliography{distance}{}

\newpage
\appendix

\section{Proof or Propositions~\ref{scalar} and~\ref{norm}}

\subsection{Outline of the proof and similarities with~\cite{Bor15}}

The main arguments of the proof rely on the study of three quantities:
\begin{enumerate}
  \item a multi-type branching process $Z_t$,
  \item a similar process based on exploring the neighbourhood of a vertex $v$ in $G$, named
    $Y_t(v)$,
  \item the actual vectors we're aiming to study, $B\l\chi_k$.
\end{enumerate}

When the $\ell$-neighbourhood of $v$ is cycle-free, we have that $B\l\chi_k = \langle \phi_k,
Y_t(v)\rangle$ for $k\in[r_0]$; and there is a coupling between the laws of $Z_t$ and $Y_t(v)$ for
almost every $v$, which allows us to translate results on $Z_t$ to results on $B\l\chi_k$.

The proof in~\cite{Bor15} studies the matrix $B^\ell$, where $B$ is the non-backtracking matrix;
$B^\ell_{ij}$ therefore counts the number of non-backtracking walks between $i$ and $j$. When the
$\ell$-neighbourhood of $i$ is tree-like, $(B\l \chi_k)_i = (B^\ell \bar\chi_k)i$, where
$\bar\chi_k$ is a similarly defined vector; most of the results from~\cite{Bor15} can therefore be
applied to this setting without further work. We will simply lay out the main steps of the proof,
highlighting the main differences with~\cite{Bor15} when necessary. 

\subsection{Local structure of $G$}

For an integer $t \geq 0$, we introduce the vector $Y_t(v) = (Y_t(v)(i))_{i\in[r]}$, where
\[ Y_t(v)(i) = \left|\Set{w\in V \given d(v, w) = t, \sigma(w) = i} \right|.\]  

The proof of our first proposition, although quite lengthy, is completely identical to its equivalent
in~\cite{Bor15}; we therefore omit it.

\begin{proposition}\label{nkl}
  Let $\ell \sim \kappa\log_\alpha(n)$ with $\kappa < 1/8$; then, for all $\gamma < 1/2$:

  \begin{enumerate}
    \item for any $k\in[r_0]$, there exists $\rho_k > 0$ such that in probability,
      \[ \frac 1 n \sum_{v\in V}{\frac{\langle \phi_k, Y_\ell(v) \rangle^2}{\mu_k^{2\ell}}} \to
        \rho_k. \]

      \item for any $j \neq k \in [r]$,
        \[ \mathbb E{\left| \frac 1 n \sum_{v\in V}{\langle \phi_j, Y_\ell(v)\rangle\langle\phi_k,
        Y_\ell(v)\rangle} \right|} = O\left(\alpha^{5\ell/2}n^{-\gamma/2}(\log(n))^{5/2}\right). \]

      \item for any $j \neq k \in [r]$,
        \[ \mathbb E{\left| \frac 1 n \sum_{v\in V}{\langle \phi_j, Y_{2\ell}(v)\rangle\langle\phi_k,
        Y_\ell(v)\rangle} \right|} = O\left(\alpha^{7\ell/2}n^{-\gamma/2}(\log(n))^{5/2}\right). \]
  \end{enumerate}
\end{proposition}

For $t\geq 0$, define $\mathcal Y_t(v) = \Set{w \in V \given d(v, w) = t}$; for $k\in[r]$, we set
\[ P_{k, \ell}(v) = \sum_{t=0}^{\ell - 1}{\sum_{w\in \mathcal Y_t(v)}{L_k(w)}}, \]
where

\[ L_k(w) = \sum_{(x, y)\in\mathcal Y_1(w) \setminus \mathcal Y_t(v), x \neq y}{\langle \phi_k,
\tilde Y_t(x)\rangle \tilde S_{\ell-t-1}(y)}, \]

$\tilde Y_t(x)$ is the equivalent of $Y_t(x)$ when all vertices in $(G, v)_t$ (i.e. vertices at
distance at most $t$ from $v$) are removed and $\tilde S_{\ell - t - 1}(y) = 
\norm{\tilde Y_{\ell - t - 1}(y)}_1$.

It can be seen from~\cite{Bor15} that when $(G, v)_{2\ell}$ is a tree, then
\[ (B\l B\l \chi_k)_v = P_{k, \ell}(v) + \chi_k(v)S_\ell(v) + \langle \phi_k, Y_{2\ell}(v) \rangle
.\]

One main difference with the proof in~\cite{Bor15} is the presence of the last term in the above
sum, as well as the fact that dealing with $B\l B\l \chi_k$ is a little more difficult. The next
proposition is an adaptation of Proposition 38 from~\cite{Bor15}, with an identical -- and thus
omitted -- proof:

\begin{proposition}\label{pkl}
  Let $\ell\sim \kappa\log_\alpha(n)$ with $\kappa < 1/10$. Then, for all $\gamma < 1/2$:

  \begin{enumerate}
    \item for all $k\in[r_0]$, there exists $\rho_k'$ such that w.h.p
      \[ \frac 1 n \sum_{v\in V}{\frac{(P_{k, \ell}(v) + \langle \phi_k, Y_{2\ell}(v)
      \rangle)^2}{\mu_k^{4\ell}}} \to \rho_k'. \]

    \item for any $j \neq k \in [r]$, for some $ c > 0$:
      \[ \frac 1 n \sum_{v  \in V}{P_{k, \ell}(v)\langle\phi_j, Y_\ell(v)\rangle} =
      O\left(\alpha^{7\ell/2}n^{-\gamma/2}(\log(n))^c\right). \]
  \end{enumerate}
\end{proposition}

\subsection{From local neighbourhoods to the matrix B\l}

For ease of notation, we define $N_{k, \ell}(v) = \langle \phi_k, Y_\ell(v) \rangle$; using the
same methods as in~\cite{Bor15}, we have the following estimates:

\begin{proposition}\label{prop-nkl}
  Let $\ell \sim\kappa\log_\alpha(n)$ with $\kappa < 1/4$. Then w.h.p:
  \[ \norm{B\l \chi_k - N_{k,l}} = o(\alpha^{\ell/2} \sqrt{n}) \textand \norm{B\l B\l \chi_k - P_{k,
  \ell} - N_{k, 2\ell}} = O(\alpha^\ell\sqrt{n}). \]
\end{proposition}

It then remains to follow the proof of Proposition 19 from~\cite{Bor15}; we simply highlight the
proof for estimation (iii) of Proposition~\ref{scalar}, since it is the only difference:

\begin{proof} (Proposition~\ref{scalar}-(iii)):
  We have by definition 
  \[ \langle \varphi_j, \zeta_k \rangle = \frac{\langle B\l \chi_j, B\l B\l \chi_k
  \rangle}{\norm{B\l\chi_j}\norm{B\l B\l\chi_k}}. \]

  But $\norm{B\l \chi_j} = \Theta(\sqrt n\mu_k^\ell)$, $\norm{B\l B\l \chi_k} = \norm{B\l\chi_k}\theta_k =
  \Theta(\sqrt{n}\mu_k^{2\ell})$ and:
  \begin{align*} 
    \left| \langle B\l \chi_j, B\l B\l \chi_k \rangle - \langle N_{j, \ell} , P_{k, \ell} + N_{k,
      2\ell} \rangle \right| & \leq \norm{N_{j, \ell}}\norm{B\l B\l \chi_k - P_{k, \ell} - N_{k,
  2\ell}}\\
  &\quad + \norm{B\l B\l \chi_k}\norm{B\l\chi_j - N_{j, \ell}} \\
  & = \widetilde O(\alpha^{4\ell}\sqrt{n}).
  \end{align*}

  Furthermore, from Propositions~\ref{nkl} and~\ref{pkl}, we get
  \[ \langle N_{j, \ell} , P_{k, \ell} + N_{k, 2\ell} \rangle = \widetilde
  O(\alpha^{7\ell/2}n^{1-\gamma/2}). \]

  This gives the desired result.
\end{proof}

\subsection{Ramanujan property of $B\l$}

In order to complete the proof of Theorem~\ref{path-expansion}, we need a control on the other
eigenvalues of $B\l$. This is covered by the following proposition:

\begin{proposition}\label{orthogonal}
  Let $H = \langle \varphi_1, \ldots, \varphi_{r_0}\rangle$, and $\ell \sim \kappa\log_\alpha(n)$ with
  $\kappa < 1/12$. Then with high probability
  \begin{equation}\label{supremum}
    \sup_{x\in H^\perp, \norm{x} = 1}{\norm{B\l x}} = \widetilde O(\alpha^{\ell/2}) .
  \end{equation}
\end{proposition}

The proof of this result relies on the following decomposition of $B\l$, whose proof can be found
in~\cite{Mas13}:

\begin{lemma}\label{decomposition}
  Matrix $B\l$ verifies the identity
  \begin{equation}
    B\l = \Delta\l + \sum_{m=1}^\ell{\Delta^{(\ell-m)}\bar A B^{(m-1)}} - \sum_{m=0}^\ell{\Gamma^{\ell,
    m}},
  \end{equation}
  for matrices $\Delta^{(j)}, \Gamma^{\ell, m}$ such that for $\ell = O(\log(n))$ and with high
  probability, for all $\varepsilon > 0$,
  \begin{align}
    \rho(\Delta^{(j)}) = \widetilde O(\alpha^{j}),\ j = 1, \ldots, \ell, \\
    \rho(\Gamma^{\ell, m}) = n^{\varepsilon - 1}\alpha^{(\ell + m)/2}, m = 1, \ldots, \ell.
  \end{align}
\end{lemma}

\noindent Here, $\bar A$ refers to the expected value of the adjacency matrix $A$ of $G$.

\vspace{0.2cm}

The next step is therefore to control $B^{(m-1)}x$ for $x \in H^\perp$; in what follows $\gamma$
will be any constant below $1/2$. We begin with the following proposition from~\cite{Bor15}:
\begin{proposition}\label{bad-vertices}
  Let $\ell \sim \kappa\log_\alpha(n)$ with $\kappa < \gamma/2$. There exists a subset $\mathcal
  B\subset V$, constants $C$ and $c$ such that w.h.p the following holds:

  \begin{enumerate}
    \item for all $i \in V\setminus \mathcal B$, $0\leq m \leq \ell$,
      \begin{align*}
        |(B^{(m)}\chi_k)_i - \mu_k^{t-\ell}(B\l\chi_k)_i| &\leq C\log(n)^c\alpha^{m/2} & \textif
        k\in[r_0], \\
        |(B^{(m)}\chi_k)_i| &\leq C\log(n)^c\alpha^{m/2} & \textif k\in[r]\setminus[r_0].
      \end{align*}

    \item for all $i\in \mathcal B$, $0 \leq m \leq \ell$ and $k\in [r]$,
      \[ |(B\l\chi_k)_i| \leq C\log(n)^c\alpha^m. \]

    \item $|\mathcal B| = \widetilde O(\alpha^\ell n^{1-\gamma})$.
  \end{enumerate}
\end{proposition}

From this, we get the following corollary:

\begin{corollary}
  Let $\ell \sim \kappa\log_\alpha(n)$ with $\kappa < \gamma/2$; then, with high probability, for $0
  \leq m \leq \ell - 1$ and $k\in [r_0]$: 
  \[ \sup_{x \perp B\l\chi_k, \norm{x} = 1}{\langle B^{(m)}\chi_k, x \rangle} = \widetilde
  O(\sqrt{n}\,\alpha^{m/2}). \]

  Additionally, for $k\in[r]\setminus[r_0]$,
  \[ \norm{B^{(m)}\chi_k} = \widetilde O(\sqrt{n}\alpha^{m/2}). \]
\end{corollary}

\begin{proof}
  Write
  \[ \langle B^{(m)}\chi_k, x \rangle = \sum_{i\in \mathcal B}{x_i (B^{(m)}\chi_k)_i} + 
    \sum_{i\notin \mathcal B}{x_i (B^{(m)}\chi_k)_i} = s_1 + s_2. \]

    Using the Cauchy-Schwarz inequality, the first sum is bounded by
    \[ |s_1| \leq \log(n)^c\alpha^m\sqrt{|\mathcal B|} \leq
    \log(n)^d\alpha^m\alpha^{\ell/2}n^{(1-\gamma)/2} = o(\sqrt{n}\alpha^{m/2}),\]

    while the second can be treated using Proposition~\ref{bad-vertices} and the fact that $\langle
    B\l\chi_k, x \rangle = 0$:
    \begin{align*}
      |s_2| &\leq \mu_k^{t-\ell}\sum_{i\in\mathcal B}{|x_i||(B\l\chi_k)_i|} + \sum_{i\notin \mathcal
      B}{|x_i||(B^{(m)}\chi_k)_i - \mu_k^{t-\ell}(B\l\chi_k)_i|} \\
      & \leq \log(n)^c \alpha^{t-\ell}\alpha^\ell\alpha^{\ell/2}n^{(1-\gamma)/2} + \log(n)^c
      \sqrt{n}\alpha^{t/2} \\
      & = \widetilde O(\sqrt{n}\,\alpha^{m/2}),
    \end{align*}
    where we used again the Cauchy-Schwarz inequality as before.

    Let now $k\in[r]\setminus[r_0]$; as before, we write
    \begin{align*}
      \norm{B^{(m)}\chi_k}^2 &= \sum_{i\in\mathcal B}{(B^{(m)}\chi_k)^2_i} + \sum_{i\notin\mathcal
      B}{(B^{(m)}\chi_k)^2_i}\\
      &\leq |\mathcal B|\log(n)^c\alpha^{2m} + n\log(n)^c\alpha^{m} \\
      &= n\log(n)^c(\alpha^{l+2m}n^{-\gamma} + \alpha^{m}) \\
      &= \widetilde O(n\alpha^m),
    \end{align*}
    and the result follows.
\end{proof}

We are now ready to prove Proposition~\ref{orthogonal}:

\begin{proof}
  Let $x \in H^\perp$ such that $\norm{x} = 1$ and the supremum in (\ref{supremum}) is reached; using the decomposition from
  Lemma~\ref{decomposition}, we have
  \[ \norm{B\l x} \leq \rho(\Delta\l) + \sum_{m=1}^\ell{\rho(\Delta^{(\ell-m)})\norm{\bar A B^{(m-1)}
  x}} + \sum_{m=1}^\ell{\rho(\Gamma^{\ell,m})}. \]

  The first and third terms are bounded by $\widetilde O(\alpha^{\ell/2})$. For the second term, we
  notice that defining the matrix $P$ by 
  \[ P = \frac 1 n \sum_{k=1}^r{\mu_k\chi_k\chi^\top_k}, \]
  we have $\bar A = P - \mathrm{diag}(P)$ since $W = \sum{\mu_k\phi_k\phi^\top_k}$.

  Therefore, for fixed $1 \leq m \leq \ell$, we have:
  \begin{align*}
    \norm{\bar A B^{(m-1)}x} &= \norm*{\sum_{k=1}^r{\mu_k\chi_k\chi^\top_kB^{(m-1)}x} -
    \mathrm{diag}(P)B^{(m-1)}x} \\
    &\leq \frac{\sup_i{W_{ii}}}n\norm{B^{(m-1)}x} + \sum_{k\in[r_0]}{\frac{\mu_k}n\norm{\chi_k\chi^\top_kB^{(m-1)}x}}
    + \sum_{k\in[r]\setminus[r_0]}{\frac{\mu_k}n\norm{\chi_k\chi^\top_k B^{(m-1)}x}} \\
    & = I + J + K.
  \end{align*}

  Notice first that $B\l_{ij} \leq 2$ for all $i, j$ by the tangle-free property, so $I = O(1)$.
  Now, for $k\in[r_0]$, we have 
  \begin{align*}
    \norm{\chi_k\chi^\top_k B^{(m-1)}x} &= \norm{\chi_k}\langle B^{(m-1)}\chi_k, x\rangle \\
                                       &\leq \widetilde O(\sqrt{n} \times \sqrt{n}\alpha^{m/2}).
  \end{align*}

  Therefore, $J = \widetilde O(\alpha^{m/2})$; finally, using the Cauchy-Schwarz inequality, we
  have for $k\in [r]\setminus [r_0]$
  \begin{align*}
    \norm{\chi_k\chi^\top_k B^{(m-1)}x} &\leq \norm{\chi_k}\norm{B^{(m-1)}\chi_k}\norm{x} \\
                                   & = \widetilde O(\sqrt{n} \times \sqrt{n}\alpha^{m/2}\times 1).
  \end{align*}

  Putting this all together, we find that for $1 \leq m \leq \ell$
  \[ \norm{\bar A B^{(m-1)}x} = \widetilde O(\alpha^{m/2}). \]

  Since $\rho(\Delta^{(\ell-m)}) = \widetilde O(\alpha^{(\ell-m)/2})$, we get $\norm{B\l x} =
  \widetilde O(\alpha^{\ell/2})$, which proves the desired result.
\end{proof}

\subsection{Proof of Proposition~\ref{norm}}

Using Proposition~\ref{orthogonal}, we are now able to prove our last result. Note that if $\kappa <
1/12$, there exists a $\gamma < 1/2$ such that $\kappa < \gamma/6$.

Let $z_k$ be the Gram-Schmidt orthonormalization of $\varphi_k$; using Lemma 9 from~\cite{Bor15},
we know that
\[ \norm{\varphi_k - z_k} = \widetilde O(\alpha^{3\ell/2}n^{-\gamma/2}), \]
and thus $z_k$ is asymptotically parallel to $\varphi_k$.

We only need a final lemma to complete our proof:
\begin{lemma}
  Assume that $\ell \sim \kappa\log_\alpha(n)$ with $\kappa < \gamma/6$. Then
  \[ \norm{\zeta_k - z_k} = \widetilde O(\theta_k^{-1}\alpha^{\ell/2}). \]
\end{lemma}

\begin{proof}
  Write 
  \[ \zeta_k = \sum_{j\in[r_0]}{\langle \zeta_k, z_j \rangle z_j} + x ,\]
  where $x \in H^\perp$.

  We have, for $j\neq k$, $\langle \zeta_k, z_j \rangle = \widetilde O(\alpha^{2\ell}n^{-\gamma/2})$ by the above
  bound of $\norm{\varphi_j - z_j}$; furthermore,
  \begin{align*} 
    \norm{x}^2 &= \langle \zeta_k, x \rangle \\
               &= \theta_k^{-1} \langle B\l\varphi_k, x \rangle \\ 
               &\leq \theta_k^{-1}\norm{B\l x} \\
               &= \widetilde O(\theta_k^{-1}\alpha^{\ell/2}) \times \norm{x}.
  \end{align*}

  Therefore, we can write
  \begin{align*}
    1 = \norm{\zeta_k}^2 &= \langle \zeta_k, z_k \rangle^2 + \sum_{j\neq k}{\langle \zeta_k, z_j
    \rangle^2} + \norm{x}^2 \\
    &= \langle \zeta_k, z_k \rangle^2 + \widetilde O(\alpha^{2\ell}n^{-\gamma/2}) + \widetilde
    O(\theta_k^{-2}\alpha^\ell) \\
    &= \langle \zeta_k, z_k \rangle^2 + \widetilde O(\theta_k^{-2} \alpha^\ell),
  \end{align*}
  since $\kappa < \gamma / 6$.

  Then,
  \[ \norm{z_k - \zeta_k}^2 = 2(1 - \langle \zeta_k, z_k \rangle)  = \widetilde
  O(\theta_k^{-2}\alpha^\ell), \]
  which yields the desired result.
\end{proof}

\begin{proof} (of Proposition~\ref{norm}): We first bound $\norm{B\l z_k - Dz_k}$ for $k\in[r_0]$.
  Notice that $Dz_k = \theta_k z_k$; this gives

  \begin{align*}
    \norm{B\l z_k - Dz_k} &\leq \norm{B\l z_k - B\l \varphi_k} + \norm{B\l \varphi_k - \theta_k z_k}
    \\
    &\leq \rho(B\l)\norm{z_k - \varphi_k} + \theta_k\norm{\zeta_k - z_k} \\
    &= O(\alpha^\ell) \times \widetilde O(\alpha^{3\ell/2}n^{-\gamma/2}) + \widetilde
    O(\alpha^{\ell/2}) \\
    &= \widetilde O(\alpha^{\ell/2}).
  \end{align*}

  Consider now $x\in \mathbb{R}^V$ such that $\norm{x} = 1$. Decomposing $x$ as $\sum{x_k z_k} + x'$ where
  $x' \in H^\perp$, we have:

  \begin{align*}
    \norm{B\l x - Dx} &\leq \sum_{k\in [r_0]}{x_k\norm{B\l z_k - Dz_k}} + \norm{B\l x' - Dx'} \\
                      &\leq \widetilde O(\alpha^{\ell/2}) + \norm{B\l x'} \\
                      &= \widetilde O(\alpha^{\ell/2}),
  \end{align*}
which completes the proof.
\end{proof}

\section{Proofs for Theorem~\ref{algorithm}}
\subsection{Proof of Lemma~\ref{lem-conv}}

We first recall a result from Kesten and Stigum: consider a multitype Galton-Watson process, where the type of the root node is distributed according to arbitrary probability vector $\nu$, and a particle of type $j \in [r]$ has a $\mathrm{Poi}(M_{ij})$ number of children of type $i$. Let $Z_t$ be the vector of population at time $t$, and $\mathcal F_t$ the natural filtration associated to $Z_t$; we have the following statement: 

\begin{lemma}\label{lem-ks}
  For each $\mu$ eigenvalue of $M$ such that $\mu^2 > \alpha$, and each eigenvector $\phi$ associated to $\mu$,
  \begin{equation}\label{eq-martingale}
  t \mapsto X(\phi, \nu, t) = \mu_k^{-t}\langle \phi_k, Z_t \rangle
  \end{equation}
  is an $\mathcal F_t$-martingale converging a.s. and in $L^2$ to a random variable with finite variance and expected value $\langle \phi_k, \nu \rangle$.
\end{lemma}

Let $\mu \neq \alpha$ be an eigenvalue of $M$ of multiplicity $d$ such that $\mu^2 > \alpha$, and $\phi^{(1)}, \dots, \phi^{(d)}$ an orthonormal basis of eigenvectors associated to $d$. We define for all $i \in [d], j\in [r]$, $X^{(i)}_j$ the limit variable of martingale (\ref{eq-martingale}), applied to $\phi = \phi^{(i)}$ and $\nu = \delta_j$. 
Similarly to previous notations, let $\phi^{(i)}$ (resp. $X^{(i)}$) be the vector $\left( \phi^{(i)}_j \right)_{j\in[r]}$ (resp.$\left( X^{(i)}_j \right)_{j\in[r]}$), and $\boldsymbol \phi $ (resp. $\mathbf X$) the (random) matrix whose columns are the $\phi^{(i)}$ (resp. the $X^{(i)}$). Recall that from Lemma~\ref{lem-ks}, the expected value of $X^{(i)}_j$ is $\phi^{(i)}_j$ for all $i, j$.

\vspace{1em}

Now, let $\xi$ be an eigenvector of $B\l$, normalized so that $\norm{\xi}^2 = n$, with associated eigenvalue $\Theta(\mu^\ell)$; as shown in the proof of Theorem~\ref{path-expansion}, there exists a vector $u \in \mathbb R^d$ such that $$\norm{\xi - \left(\langle \boldsymbol \phi u, Y_\ell(v)\rangle\right)_{v\in V}} = o(1).$$
We let
\[\phi^{(\xi)} = \boldsymbol \phi u \quad\textand \quad X^{(\xi)} = \mathbf X u.\]

From Proposition~\ref{nkl}, $u$ has norm $\Theta(1)$, and since $\mu^{-t}\langle \phi, Z_t \rangle$ (with $\nu = \delta_j$) converges to $X^{(i)}_j$ in $L^1$, $\mu^{-t}\langle \phi^{(\xi)}, Z_t \rangle$ converges to $X^{(\xi)}_j$ in $L^1$ independently of $\xi$.

\vspace{1em}

Using proposition 36 from \cite{Bor15}, we have the following:

\begin{lemma}\label{lem-mean}
  For all $i \in [r]$, we have the following convergence in $L^1$:
  \[ \frac1n \sum_{v \in V}\mathbf{1}_{\sigma(v)=i}\, \xi(v)\mathbf 1_{|\xi(v)| \leq K} \to \pi(i)\E*{X^{(\xi)}\mathbf 1_{|X^{(\xi)}| \leq K}},\]

  for all $K$ that is a continuity point of the distribution of $X_i$, and independently of $\xi$.
\end{lemma}

\begin{proof}
  We first recall the aforementioned proposition from \cite{Bor15}: we say that a function $\tau$ that takes a graph and a distinguished vertex as an argument is $\ell$-local if $\tau(G, v)$ depends only on the $\ell$-neighbourhood of $v$ in $G$. Denote by $T$ the multitype Galton-Watson tree discussed earlier, rooted at $o$, where $o$ has the distribution $\delta_\iota$ and $\iota$ has distribution $\pi$.
  \begin{lemma}
      Assume that $\tau, \psi$ are two $\ell$-local functions such that $|\tau| \leq \psi$ and $\psi$ is non-decreasing by the addition of edges. Then, if $\ell \sim \kappa\log_\alpha(n)$ with $\kappa < 1/2$, we have, for $\gamma < 1/2$:
      \begin{align*} 
        &\E*{\left|\frac1n\sum_{v\in V}\tau(G, v) - \E*{\tau(T, o)}\right|} \\
        &\qquad \leq c \frac{\alpha^{\ell/2}\sqrt{\log(n)}}{n^{\gamma/2}}\left( \left(\E{\max_{v\in V}\psi^4(G, v)}\right)^{1/4} \vee \left(\E*{\psi^2(T, o)}\right)^{1/2} \right)
      \end{align*}
  \end{lemma}
  We now apply this lemma with $\tau(G, v) = \mathbf{1}_{\sigma(v) = i} \, \langle \phi^{(\xi)}, Y_\ell(v)\rangle\mathbf 1_{|\langle \phi^{(\xi)}, Y_\ell(v)\rangle| \leq K}$ where $Y_\ell$ is defined in Proposition~\ref{nkl}. \\
  We can set $\psi(G, v) = K$, and by Lemma~\ref{lem-ks} and the subsequent analysis, we have
  \[ \E*{\tau(T, o)} \to \pi(i)\,\E*{X^{(\xi)}_i \mathbf 1_{|X^{(\xi)}_i| \leq K}} \]
  independently of $\xi$.

  Now, by definition, we have $\norm{\xi - \left(\langle \phi^{(\xi)}, Y_\ell(v)\rangle\right)_{v\in V}} = o(1)$ (again independently of $\xi$). By the Cauchy-Schwarz inequality, we deduce that
  \[ \frac1n\sum_{v \in V}\left|\xi(v) - \langle \phi^{(\xi)}, Y_\ell(v)\rangle\right| = o(1) \]
  as well, and the lemma follows if $K$ is a continuity point of $X^{(\xi)}_i$.
\end{proof}

It now remains to prove the desired properties of the $X^{(\xi)}_i$; first, since $\mu \neq \alpha$, then $\phi^{(\xi)}$ is orthogonal to the all-one vector and as such
\[ \sum_{i\in [r]}{\E{X_i}} = \sum_{i\in [r]}{\phi^{(\xi)}_i} = 0 \]

Moreover, $\norm{\phi^{(\xi)}}^2 = \norm{u}^2 = \Theta(1)$, which proves the second assertion.

Finally, let $\eta > 0$;  since the $X^{(i)}_j$ all have finite variance, there exists a constant $K > 0$ such that $\Pb{|X^{(i)}_j| \leq K' } \geq 1 - \eta$ for all $i, j$ and thus
\[ \Pb*{\norm{\mathbf X}_\infty \leq K'} \geq 1 - dr\eta. \]
Using the equivalence of norms, we find
\[ \Pb*{\forall i, \ \norm{X^{(i)}}_2^2 \leq r K'^2} \geq 1 - dr\eta \]
which implies (since $X^{(\xi)} = \mathbf Xu$)
\[ \Pb*{\norm{X^{(\xi)}}^2_2 \leq r \norm{u}_2^2 K'^2} \geq 1 - dr\eta. \]
Using again norm equivalence yields finally, for $K = \sqrt{r} \norm{u} K'$,
\begin{equation}\label{eq-m}
  \Pb*{\norm{X^{(\xi)}}_\infty \leq K} \geq 1 - dr\eta.
\end{equation}

Now, we have, for all $i$,
\begin{align*} 
  \left|\E{X^{(\xi)}_i \mathbf{1}_{|X^{(\xi)}_i| \leq K}} - \E{X^{(\xi)}_i}\right| &= \E*{X^{(\xi)_i}\mathbf 1_{|X_i^{(\xi)}| > K}}\\
  &\leq \sqrt{\E*{(X_i^{(\xi)})^2} \Pb*{|X^{(\xi)}_i| \geq K}} \\
  &\leq \sqrt{\E*{(X_i^{(\xi)})^2}} \cdot \sqrt{dr\eta} \\
\end{align*}

But by Doob's Theorem, $\E*{(X^{(\xi)}_i)^2}$ is finite so choosing $\eta$ accordingly yields the last inequality of Lemma~\ref{lem-conv}.

\subsection{Proof of limit (\ref{eq-proportion})}

For each $v \in V$, we define $I_v$ to be the random variable equal to 1 if $v$ is assigned to $I^+$, and 0 otherwise. Conditionnally to $\xi$, it is straightforward to see that
\[ I_v \sim \mathrm{Ber}(q_v) \quad \text{with} \quad q_v = \frac12 + \frac1{2K}\xi(v)\mathbf{1}_{|\xi(v)| \leq K} \]

Now, let $P_i = \left(\left|\Set{v \in I^+ \given \sigma(v) = i}\right|\right)/n$; by definition, 
$$ P_i = \frac1n\sum_{v \in V} \mathbf{1}_{\sigma(v) = i}\,I_v $$.

We therefore have
\[ \E*{P_i \given \xi} = \frac1n\sum_{v \in V} \mathbf{1}_{\sigma(v) = i}\,\left(\frac12 + \frac1{2K}\xi(v)\mathbf{1}_{|\xi(v)| \leq M} \right), \quad \mathrm{Var}(P_i) \leq \frac1{4n} \]

and thus with high probability, independently of $\xi$,
\begin{align}
  P_i &= \E*{P_i \given \xi} + n^{-1/3} \\
      &\to \pi(i)\left(\frac12 + \frac1{2K}\, \E{X^{(\xi)}_i \mathbf 1_{|X^{(\xi)}_i| \leq M}}\right) = \pi(i)\, \tilde p_i
\end{align}

where the convergence speed is independent from $\xi$.

\subsection{Explicit bounds on $K$}

In this section, the goal is to perform a more precise analysis of the limit variables $X^{(\xi)}_i$, and to leverage this analysis to obtain an explicit value for $K$ in Algorithm~\ref{algo-reconstruction}. For simplicity, we will assume that $\pi \equiv 1/r$ throughout this section, although most of the results hold for general $\pi$. We begin with a small lemma:

\begin{lemma}\label{lem-cumulant}
  Let $\phi$ be a normed eigenvector of $M$ associated to an eigenvalue $\alpha > \mu > \sqrt{\alpha}$, and denote by $X^{(\phi)}_i$ the limit random variables of Lemma~\ref{lem-ks}. Let $c^{(\phi), j}$ (resp. $m^{(\phi), j}$) the vector of the $j$-th cumulants (resp. moments) of the $X^{(\phi)}_i$. We then have the following relation, for all $j \in \mathbb N$:

  \[ c^{(\phi), j} = \frac M{\mu^j} m^{(\phi), j} \]
\end{lemma}

By definition, we have $c^{(\phi), 1} = m^{(\phi), 1} = \phi$, and we have the following corollary for $c^{(\phi), 2}$ and $m^{(\phi), 2}$:
\begin{corollary}\label{cor-variance}
  We denote by $\phi^2$ the vector whose coordinates are the $\phi_i^2$. Then
  \[ c^{(\phi), 2} = \left( I - \frac{M}{\mu^2} \right)^{-1}\frac{M}{\mu^2}\, \phi^2 \quad \textand \quad m^{(\phi), 2} = \left( I - \frac{M}{\mu^2} \right)^{-1}\phi^2. \]

  As a result, we have
  \[ \sum_{i\in [r]}{\mathrm{Var}(X^{(\phi)}_i)} = \frac1{\tau - 1} \quad \textand \quad \sum_{i\in [r]}{\E{(X_i^{(\phi)})^2}} = \frac{\tau}{\tau - 1}, \]
  where $\tau = \mu_2^2 / \alpha$.
\end{corollary}

\begin{proof}(of Corollary~\ref{cor-variance}).
  The first part is an easy calculation, observing that $c^{(\phi), 2} = m^{(\phi), 2} - \phi^2$.

  For the second part, since the all-one vector $e$ is an eigenvector of $M$ associated to the eigenvalue $\alpha$, we have:
  \begin{align*}
    \sum_{i\in [r]}{\mathrm{Var}(X^{(\phi)}_i)} &= \langle e , c^{(\phi), 2} \rangle \\
    &= e^\top \left( I - \frac{M}{\mu^2} \right)^{-1}\frac{M}{\mu^2}\, \phi^2 \\
    &= \frac{\alpha/\mu^2}{1 - \alpha/\mu^2} e^\top\, \phi^2 \\
    &= \frac{1}{\tau - 1},
  \end{align*}
  and a similar calculation yields the second identity.
\end{proof}

It now remains to prove Lemma~\ref{lem-cumulant}:

\begin{proof}(of Lemma~\ref{lem-cumulant}). Using the Galton-Watson tree definition (and going one step down into the tree), we have the following characterization for the variables $X^{(\phi)}_i$:
  \[ X^{(\phi)}_i = \frac1\mu\sum_{j \in [r]} \sum_{k = 1}^{\mathrm{Poi(M_{ij})}}{ X^{(\phi)}_{j, k} }, \]
where the $X^{(\phi)}_{j, k}$ are independent copies of $X^{(\phi)}_j$ for all $k$. Applying the Laplace transform (denoted by $\psi^{(\phi)}_i$) and taking the logarithm on both sides, we find that for all $t \in \mathbb R$,

\[ \log(\psi^{(\phi)}_i(t)) = \sum_{j\in[r]}{M_{ij}\left(\psi^{(\phi)}_j\left(\frac t \mu\right) - 1\right)} \]

Now, the $k$-th Taylor coefficient of the LHS is $c^{(\phi), k}_i / k!$, and the one on the RHS is 
$$\sum_{j\in[r]}{M_{ij} \frac{m^{(\phi), k}_j}{k!\, \mu^k}} = \frac1{\mu^k \, k!} \left[ M m^{(\phi), k} \right]_i, $$
which completes the proof.
\end{proof}

We now can prove our first result on the vector $u$ defined before:
\begin{lemma}\label{lem-v}
  Let $\mu$, $\xi$ and $u$ be defined as in the proof of Lemma~\ref{lem-conv}. Then we have 
  \[ \norm{u}^2 = r(\tau - 1) + o(1) \]
\end{lemma}

\begin{proof}
  From Lemma~\ref{nkl}, we know that for each $i \in [d]$,
  \[ \norm*{\left(\langle \phi^{(i)}, Y_\ell(v) \rangle \right)_{v \in V}}^2 = n(\rho^{(i)}+ o(1)) \quad \text{where} \quad \rho^{(i)} = \sum_{i \in [r]}{\pi(i)\E*{(X^{(i)})^2}} \]

  But since $\pi \equiv 1/r$, we know from Corollary~\ref{cor-variance} that 
  $$ \rho^{(i)} = \rho := \frac 1{r(\tau - 1)} $$

  But the vectors $\left(\langle \phi^{(i)}, Y_\ell(v) \rangle \right)_{v \in V}$ are asymptotically orthogonal, and thus
  \[ n = \norm{\xi}^2 = (\norm{v}^2 + o(1)) \cdot n \cdot (\rho + o(1)), \]
  which yields the desired result.
\end{proof}

Now, we are ready to prove some bounds for $K$; the main step is the following Markov bound on $(X^{(i)})^2$:

\begin{lemma}\label{lem-markov}
  Let $\eta > 0$; then, for all $i \in [d]$, $j \in [r]$, 
  \[ \Pb*{|X^{(i)}_j| \leq \sqrt{\frac{\tau}{\eta(\tau - 1)}}\, } \geq 1 - \eta \]
\end{lemma}

\begin{proof}
  For all $C > 0$, we have by Markov's inequality

  \[ \Pb*{|X^{(i)}_j| \geq C} \leq \frac{\E{(X^{(i)}_j)^2}}{C^2} \leq \frac{\tau}{C^2(\tau - 1)}, \]
  where we bounded $\E{(X^{(i)}_j)^2}$ by the sum of all $\E{(X^{(i)}_k)^2}$. The lemma then follows easily.
\end{proof}

Now, we have to unravel the calculations done in the proof for Lemma~\ref{lem-conv}; let $\varepsilon < 0$. By the same bound as above (as well as the fact that the $\phi^{(i)}$ are orthogonal), we have 
$$\E*{(X^{(\xi)}_i)^2} \leq \norm{u}^2 \frac\tau{\tau - 1} = r \tau + o(1)$$

Therefore, an asymptotically good choice of $\eta$ is
\[ \eta = \frac{\varepsilon^2}{r^2\,d\,\tau}, \]
which yields a value for $K'$ of
\[ K' = \sqrt{\frac{\tau}{\eta(\tau - 1)}} = \frac{r\,\tau}{\varepsilon}\sqrt{\frac{d}{\tau - 1}} \]

Finally, the bound for $K$ becomes 
$$ K = \sqrt{r}\norm{u} K' = \frac{r\,\tau}{\varepsilon}\sqrt{r\,d\,\tau} $$

We thus need to find a sufficient value for $\epsilon$; recall that
\[ \sum_{i\in [r]} \E*{X^{(\xi)}_i}^2 = \norm{u}^2 = r (\tau - 1) + o(1) \quad \textand \quad \sum_{i\in [r]} \E*{X^{(\xi)}_i} = 0, \]
and therefore there exists some values for $i, j$ such that
$$ \left|\E*{X^{(\xi)}_i}^2 - \E*{X^{(\xi)}_j}^2 \right| \geq \sqrt{r(\tau - 1)} + o(1). $$

A sufficient choice of $\varepsilon$ is thus $\sqrt{r(\tau - 1)}/4$, which yields an explicit value for $K$:

$$ K = \frac{r\,\tau}{\varepsilon}\sqrt{r\,d\,\tau}= r\,\tau \sqrt{d\,\frac{\tau}{\tau - 1}} $$

\section{Proof of Lemma~\ref{zero-one}}

We first recall some results about the neighbourhoods of vertices, whose proofs can be found
in~\cite{Bor15}:

\begin{lemma}\label{vertex-size}
  For a vertex $i$, define $S_t(i)$ as the number of vertices at distance $t$ of $i$.

  Then there exist constants $C$ and $\varepsilon > 0$ such that with probability $1 -
  O(n^{-\varepsilon})$, for all $i\in\{1, \dots n\}$ and $\ell = O(\log(n))$:
  \begin{equation}
    S_t(i) \leq C\cdot \log(n) \cdot \alpha^t,\quad t\in\{1, \dots, \ell\}.
  \end{equation}

  On the other hand, with high probability, when $\ell = \kappa\log_\alpha(n)$ with $\kappa < 1/2$:
  \begin{equation}
    \sum_{i=1}^n{S_\ell(i)^2} = \Theta(n\alpha^{2\ell}).
  \end{equation}
\end{lemma}

Additionnally, a result about the almost tree-like structure of vertex neighbourhoods:

\begin{lemma}\label{nocycles}
  Assume $\ell = \kappa\log(n)$, with $\kappa\log(\alpha) < 1/4$. Then with high probability
  no node $i$ has more than one edge cycle in its $\ell$-neighbourhood; we say that $G$ is
  $\ell$-tangle-free.
\end{lemma}

Using those results, we are now able to prove Lemma~\ref{zero-one}:

\vspace{1em}

\begin{proof}
  From Lemma~\ref{nocycles}, we can deduce that if $d(i,j) \leq \ell$, there are at most two distinct paths
  between $i$ and $j$. Therefore, $B\l_{ij} \leq 2$ for all $i,j$.

  Additionally, if $D\l_{ij} = 1$, then there is a self-avoiding path of length $\ell$ between $i$
  and $j$, and thus $B\l_{ij} = 1$, so $\Delta\l_{ij} \geq 0$ for all $i,j$.

  Finally, assume that there exists a pair $i,j$ such that $D\l_{ij} = 0$ and $B\l_{ij} = 2$; then
  there are two paths of length $\ell$ between $i$ and $j$ and $d(i,j) < \ell$ so there is also a
  path of length less than $\ell$. This contradicts Lemma~\ref{nocycles}.

  Consider now two vertices $i$ and $j$ such that $\Delta\l_{ij} = 1$, there are two possibilities:
  \begin{enumerate}
    \item $D\l_{ij} = 0$ and $B\l_{ij} > 0$: then $d(i, j) < \ell$ and there is a
    path of length $<\ell$ and at least a path of length $\ell$ between $i$
    and $j$.
    \item $D\l_{ij} = 1$ and $B\l_{ij} > 1$: then there are at least two
      paths of length $\ell$ between $i$ and $j$.
  \end{enumerate}

  In both cases, there are at least two paths of length at most $\ell$
  connecting $i$ and $j$, which implies the statement of the lemma.
\end{proof}

\section{Proof of Lemma~\ref{cycle-decomposition}}
\begin{proof}
  (i) is obvious since $d(i, \mathcal{C}) + d(j, \mathcal{C}) \leq \ell$ implies $d(i, \mathcal{C})\leq
  \ell$.

  For (ii), note first that $V_\mathcal{C}$ and $V_{\mathcal{C}'}$ are disjoint for $\mathcal{C} \neq
  \mathcal{C}'$:
  if $i\in V_\mathcal{C}\cap V_{\mathcal{C}'}$, then $\mathcal{C}$ and $\mathcal{C}'$ are in the $\ell$-neighbourhood of
  $i$, which contradicts Lemma \ref{nocycles}.

  Let $\pi_\mathcal{C}$ be the projection on $V_\mathcal{C}$ for all $\mathcal{C}$; the $\pi_\mathcal{C}$ are mutually orthogonal
  and for a vector $v$, we have:
  \begin{align}
    v^\top P\l v = \sum_{\mathcal{C}}{v^\top \pi_\mathcal{C} P\l \pi_\mathcal{C} v} & = \sum_{\mathcal{C}}{(\pi_\mathcal{C}
    v)^\top P\l_\mathcal{C}(\pi_\mathcal{C} v)} \\
    & \leq \sum_\mathcal{C}{\rho(P\l_\mathcal{C}) \cdot \norm{\pi_\mathcal{C} v}^2} \\
    & \leq \max_\mathcal{C}{\rho(P\l_\mathcal{C})} \cdot \sum_\mathcal{C}{\norm{\pi_\mathcal{C} v}^2}. \label{orth-decomp}
  \end{align}

  On the other hand,
  \begin{equation}\label{norm-decomp}
    \norm{v}^2 \geq \sum_\mathcal{C}{\norm{\pi_\mathcal{C} v}^2}.
  \end{equation}

  Combining inequalities~(\ref{orth-decomp}) and~(\ref{norm-decomp}) yields $\rho(P\l) \leq
  \max_\mathcal{C}{\rho(P\l_\mathcal{C})}$; the reverse inequality comes from the decomposition $P\l =
  \sum{P\l_\mathcal{C}}$.

\end{proof}

\section{Proof of Proposition~\ref{general-spectral}}

In the same vein as Lemma~\ref{vertex-size}, for a vertex set $\mathcal X$, define $S_t(\mathcal X)$ as the number of vertices
at distance $t$ of $\mathcal X$. By taking the union on all vertices of $\mathcal X$, we easily get the following corollary:

\begin{corollary}\label{set-size}
  For the same constants $C$ and $\varepsilon$ as above, with probability $1 - O(n^{-\varepsilon})$, we have for
  all vertex subsets $\mathcal X \in \mathcal P(\Set{1, \ldots, n})$ and $\ell = O(\log(n))$:
  \[S_t(\mathcal X) \leq C\cdot |\mathcal X|\log(n) \cdot \alpha^t,\quad t\in\{1, \dots,
  \ell\}. \]
\end{corollary}

We are now able to prove Proposition~\ref{general-spectral}:

\begin{proof}
  Let $\mathcal{K}$ be the modified vertex set, and consider vertices $i$ and $j$ such
  that $D\l_{ij} \neq \widetilde D\l_{ij}$. Then we have one of four possibilities:

  \begin{enumerate}
    \item $\tilde d(i, j) = \ell$ and $d(i,j) < \ell$
    \item $\tilde d(i, j) > \ell$ and $d(i,j) = \ell$
    \item $\tilde d(i, j) = \ell$ and $d(i,j) > \ell$
    \item $\tilde d(i, j) < \ell$ and $d(i,j) = \ell$
  \end{enumerate}

  In cases (i) and (ii), there is a path between $i$ and $j$ in $G$ through $\mathcal{K}$
  of length at most $\ell$, and in cases (iii) and (iv) there is a path between
  $i$ and $j$ in $\tilde G$ through $\mathcal{K}$. Therefore, in all cases, we have 
  that 
  \[ d(i, \mathcal{K}) + d(j, \mathcal{K}) \leq \ell.\]

  Write $|\widetilde D\l - D\l|$ for the matrix whose $(i,j)$ coefficient is $|\widetilde D\l_{ij} -
  D\l_{ij}|$, and $P_\mathcal{K}$ for the matrix such that $P_{\mathcal{K}, ij} = \ind{d(i, \mathcal{K}) + d(j, \mathcal{K}) \leq
  \ell}$; the previous analysis and the Perron-Frobenius theorem imply that
  \begin{equation}
    \rho(\widetilde D\l - D\l) \leq \rho(|\widetilde D\l - D\l|) \leq \rho(P_\mathcal{K}).
  \end{equation}

  We can then perform the same analysis as in the proof of Proposition~\ref{spectral-radius} to find
  that the spectral radius of $P_\mathcal{K}$ is the same as that of
   \[ Q_\mathcal{K} = \begin{pmatrix}
     S_0 & \sqrt{S_0S_1} & \cdots & \sqrt{S_0S_{\ell-1}} & \sqrt{S_0S_\ell} \\
     \sqrt{S_0S_1} & S_1 & \cdots & \sqrt{S_1S_{\ell-1}} & 0 \\
     \vdots & \vdots & \iddots & \vdots & \vdots \\
     \sqrt{S_0S_{\ell-1}} & \sqrt{S_1S_{\ell-1}} & \cdots & 0 & 0 \\
     \sqrt{S_0S_{\ell}} & 0 & \cdots & 0 & 0 \\
   \end{pmatrix},\]

   where we write $S_t$ instead of $S_t(\mathcal{K})$ for ease of notation. 

   Corollary~\ref{set-size} then gives $S_t(\mathcal{K}) = O(\alpha^t\log(n)|\mathcal{K}|) =
   o(\alpha^t\tau^{\ell/2})$, and the same calculation as in Proposition~\ref{spectral-radius} yields:
   \begin{equation}
     \rho(Q_\mathcal{K}) = o(\alpha^{\ell/2}\tau^{\ell/2}) = o(\mu_2^\ell),
   \end{equation} 
   and the theorem follows.

\end{proof}

\section{Proof of Theorem~\ref{optimality}}

In order to prove Theorem~\ref{optimality}, we need to show that the controls in the proof of
Theorem~\ref{generalpert} are actually sharp. We begin with the following lemma, which comes from
the fact that $\ell$-neighbourhoods of the vertices of $G$ are roughly of the same size:

\begin{lemma}\label{large-neighbourhood}
   Assume that $\gamma = \Theta(\tau^{\ell/2})$. Then there exists a set of vertices $\mathcal{K}$ of size
   $\gamma$ such that:
   \begin{equation}
    S_\ell(\mathcal{K}) = \Omega(\alpha^\ell\cdot\gamma). 
   \end{equation}
\end{lemma}

\begin{proof}
  Let $\varepsilon > 0$ to be determined later, $S$ be the set consisting of the $n^{1-\varepsilon}$ vertices $i$ 
  with the largest values $S_\ell(i)$; we first show that, for all $i\in S$
  \begin{equation}\label{large-vertex-neighbourhood}
    S_\ell(i) = \Theta(\alpha^\ell).
  \end{equation}

  Indeed, from Lemma~\ref{vertex-size}, we have the folowing inequalities:
  \begin{equation}
    Kn\alpha^{2\ell} \leq \sum_{i=1}^n{S_\ell(i)^2} \leq n\min_{i\in S} S_l(i)^2 + 
    |S|(C\log(n)\alpha^{\ell})^2,
  \end{equation}
  and the second term is negligible before the two others, which implies
  (\ref{large-vertex-neighbourhood}).

  We then build a set $\mathcal{K}$ of size $\gamma$ as follows: begin with any member of $S$, and at each step 
  add a vertex $x$ such that $d(x, \mathcal{K}) > 2\ell$. This is possible as long as the
  $2\ell$-neighbourhood of $\mathcal{K}$ does not cover $S$, i.e.\ as long as:
  \begin{equation}
    \gamma \cdot C\log(n)\alpha^{2\ell} < n^{1-\varepsilon}.
  \end{equation}

  But the LHS of this inequality is bounded by $C\log(n)n^{3/4}$, so this condition is satisfied as
  long as $\varepsilon < 1/4$.

  By this construction, the vertices of $\mathcal{K}$ have $\ell$-neighbourhoods that are pairwise disjoint,
  so by equation (\ref{large-neighbourhood}) we have:
  \begin{equation}
    S_\ell(\mathcal{K}) = \Omega(\alpha^{\ell}\times\gamma).
  \end{equation}

  \end{proof}

Consider now the vector $v$ such that:
  \begin{equation}
    v_i = \begin{cases}
      \gamma^{-1/2} & \textif i\in\mathcal{K} \\
      S_\ell(\mathcal{K})^{-1/2} & \textif d(i, \mathcal{K}) = \ell \\
      0 & \otherwise.
    \end{cases}
  \end{equation}

The aim is to show the following equalities:
    \begin{equation}\label{eq-orthogonal} 
        \frac{v^\top D\l v}{\norm{v}^2} = \Omega(\mu_2^\ell) \quad \textand \quad \langle v, B\l \chi_k \rangle = o(\norm{v}\norm{B\l \chi_k})\ \forall k \in [r_0]
    \end{equation}
Indeed, Theorem \ref{optimality} will then follow from a simple application of Courant-Fisher's Theorem.

\vspace{1em}

\begin{proof}(of Eq. (\ref{eq-orthogonal}).
  We notice that $\norm{v}^2 = 2$; furthermore:
  \begin{align*}
    v^\top D\l v &= \sum_{i,j}{v_i D\l_{ij}v_j} \\
    &\geq 2\sum_{i\in S_\ell(\mathcal{K})}{\sum_{j\in\mathcal{K}}{v_iv_j}} \\
    &= 2\gamma S_\ell(\mathcal{K}) \gamma^{-1/2}S_\ell(\mathcal{K})^{-1/2} \\
    &= 2\sqrt{\gamma S_\ell(\mathcal{K})} \\
    & = \Omega(\mu_2^\ell),
  \end{align*}

  which proves the first inequality.

  It remains then to prove that $v$ is asymptotically orthogonal to $B\l \chi_k$ for $k\in [r_0]$: 
  noticing that $v_i \leq 1$ for all $i$ and $\norm{v}_0 = \gamma + S_\ell(\mathcal{K})$, we find, using
  Corollary~\ref{set-size}:
  \begin{align*}
    \langle v, B\l \chi_k \rangle &\leq (\gamma + S_\ell(\mathcal{K}))\cdot \norm{B\l\chi_k}_\infty \\
    &\leq (\gamma + S_\ell(\mathcal{K}))\cdot \widetilde O(\alpha^\ell) \\
    &= \widetilde O(\gamma \alpha^{2\ell}) \\
    &= o(\sqrt{n}\mu_k^\ell) \quad \text{since } \kappa < 1/4 \\
    &= o(\norm{v}\norm{B\l \chi_k}),
  \end{align*}
  where we used part (ii) of proposition \ref{bad-vertices} to bound
  $\norm{B\l\chi_k}_\infty$.
\end{proof}

\end{document}